\documentclass{article}

\usepackage{microtype}
\usepackage{graphicx}
\usepackage{subfigure}
\usepackage{booktabs} 
\usepackage{amsmath,amsthm,epsfig,amsfonts,amssymb}
\usepackage[latin1]{inputenc}
\usepackage{graphics}
\usepackage{amstext}
\usepackage{color}
\usepackage{float}
\usepackage{natbib}
\usepackage{appendix}
\usepackage{caption}

\usepackage{multicol}
\usepackage[export]{adjustbox}

\usepackage{hyperref}

\usepackage[accepted]{icml2019}

\icmltitlerunning{Anomaly Detection Model for Imbalanced Datasets}

\newtheorem{proposition}{Proposition}

\theoremstyle{definition}

\begin{document}

\twocolumn[
\icmltitle{Anomaly Detection Model for Imbalanced Datasets}





\begin{icmlauthorlist}
\icmlauthor{R\'egis Houssou}{heig}
\icmlauthor{Stephan Robert-Nicoud}{heig}
\end{icmlauthorlist}

\icmlaffiliation{heig}{School of Management and Engineering Vaud (HEIG-VD), University of Applied Sciences and Arts of Western Switzerland (HES-SO), CH-1400 Yverdon-les-Bains}

\icmlcorrespondingauthor{R\'egis Houssou}{regis.houssou@heig-vd.ch}

\icmlkeywords{Machine Learning, ICML}

\vskip 0.3in
]




\printAffiliationsAndNotice{} 

\begin{abstract}
This paper proposes a method to detect bank frauds using a mixed approach combining a stochastic intensity model with the probability of fraud observed on transactions. It is a dynamic unsupervised approach which is able to predict financial frauds. The fraud prediction probability on the financial transaction is derived as a function of the dynamic intensities. In this context, the Kalman filter method is proposed to estimate the dynamic intensities. The application of our methodology to financial datasets shows a better predictive power in higher imbalanced data compared to other intensity-based models.
\end{abstract}

\section{Introduction}
\label{Introduction}
Financial fraud is growing exponentially, especially because of the large sums involved. It is an issue that has wide consequences in both the finance industry and the daily life. Fraud can reduce confidence in industry, destabilise economies, and affect people's cost of living. However, as a first step, banks and financial institutions have approached the detection of fraud using traditional approaches based on manual techniques such as auditing, which are inefficient and unreliable due to the complexities  associated with the problem. This is a very relevant problem that demands the attention of communities such as machine learning and data science where the solution of the problem can be automated, and evolve the detection of fraud towards methods using adaptive rules to tighten the mesh of the network.

 
The machine learning models work with many parameters and are much more efficient at finding subtle correlations in the data, which can be masked by an expert system or by human criticism, \cite{Dyz}. The large volume of transactional data and client data readily available in the financial services industry makes it an ideal tool for the application of complex machine learning algorithms. In addition to learning from known models, machine learning can go further and learn new models without human operation. This allows models to adapt over time to discover previously unknown patterns or to identify new tactics that can be used by fraudsters. In fact, the development of conventional machine learning algorithms has led them to solve some specific problems, one of the most important features of which is that the distribution of data is generally balanced, unlike financial fraud, which is not balanced. Most standard classifiers such as decision trees and neural networks assume that learning samples are evenly distributed among different classes. However, in many real-world applications, the ratio of the minority class is very small( 1:100, 1:1000 or can be exceeded at 1:10000).  Due to the lack of data, few samples of the minority learning class tend to be falsely detected by the classifiers and the decision limit is therefore far from correct. Numerous research works in machine learning has been proposed to solve the problem of data imbalance; \cite{Garcia}, \cite{Galar}, \cite{Kraw}, \cite{Abra}, etc. However, most of these algorithms suffer from certain limitations in real-world applications, such as the loss of usual information, classification cost, excessive time, and adjustments, see \cite{Abra}.

\cite{regis} investigated the problem of fraud detection in imbalanced data using the Poisson process. They defined the fraud times as the jump times of the Poisson process with intensity that describes the instantaneous rate of fraud. They showed how to estimate the intensity function in deterministic form and how to predict fraud events. The comparison of their methodology to other baseline approaches shows a better predicting power especially in very imbalanced dataset. However, their approach suffers from some limitations such as - The reduced form of the model in the sense that the fraud detection depends uniquely of the intensity's parameters; the model does not look inside the subtle correlations in the data. - The deterministic form of the intensity meaning that the intensity is a function of times, so it is predictable. In addition, their model is a supervised approach for which the lack of labelled data constitutes the main constraint in fraud detection. 

In this paper, we address these issues by considering a stochastic process for the fraud intensity; in other words, the intensity is a function of time and for a fixed time it is a random variable. In contrast to \cite{regis}, the instantaneous rate of fraud is no longer predictable and this is more realistic. For the calibration purpose, we also consider the posterior probabilities of fraud observed on each transaction; we suppose these probabilities reflect the likelihood of fraud in the dataset and they take into account the hidden correlations between the features. Our approach is a mixed approach combining the stochastic intensity with the probability of fraud observed on transaction. For the intensity's model, we focussed on the Cox-Ingersoll-Ross (CIR) model assuming that the trend of the fraud intensity is mean-reverting and the fraud intensity is always positive. Another main advantage of choosing this process is that we can derive a closed form solution of the prediction probability of fraud. As the intensity is unobservable variable, we propose to estimate its values by the Kalman-Filter method where the intensity is updated by the probability of fraud observed on transaction. Finally, our model is unsupervised approach in the sense that labelled data with examples of fraud are not need for detecting fraud events. 

However, a lot of research based on the Kalman filter has been done in the financial fields such as the interest rate models, the volatility models, the pricing of the defaultable bonds; see \cite{simon}, \cite{Duan}, \cite{ray}, \cite{vo}, etc...

The rest of the paper is organized as follows. Section II focusses on the fraud detection in the context of the stochastic intensity; the Cox-Ingersoll-Ross (CIR) intensity model is investigated. The prediction probability of fraud is derived and the estimation process of the intensity is discussed. In the section III, the model is applied to financial datasets and the results are presented. The dataset was provided by NetGuardians \footnote{https://netguardians.ch}, a swiss company which develops solutions for banks to proactively prevent fraud.

\section{Fraud detection with stochastic intensity}
\subsection{Cox-Ingersoll-Ross (CIR) process for intensity}
Consider a financial institution such as a bank, an insurance company, a trading company, etc. and information about its clients. We are interested in the occurrence of fraud in client transactions for such an institution. The fraud event is then defined as a rare event occurring at a random time and resulting in significant financial losses for the client and the financial institution. Let define $(\Omega, \mathcal {F}, \mathbb {F}, \mathbb {P}) $, the filtered probability space with $\Omega$ denotes the possible states of the world, $\mathcal{F}$ is the $\sigma$-algebra, $\mathbb{F}=(\mathcal{F}_{t})_{t\geq 0}$ is the filtration with $\mathcal{F}_{t}$ contains all information up to time $t$ and $\mathcal{F}_{T}=\mathcal{F}$. $\mathbb{P}$ is the probability measure describing the likelihood of certain events. We  denote by $\lambda$, the intensity that represents the expected number of fraud events per unit of time. As in \cite{regis}, one assumes that $\lambda$ is a non-negative process. In addition, we consider that the intensity is stochastic and follows the Cox-Ingersoll-Ross (CIR) process

\begin{equation}
d\lambda_{t}= \kappa (\theta-\lambda_{t})dt+\sigma\sqrt{\lambda_{t}}dB_{t}
\label{lcir}
\end{equation}

where $\kappa$, $\theta$ and $\sigma$ are positive constant; $\kappa$ represents the rate of mean reverting, $\theta$ is the long run average, $\sigma$ is the volatility of the intensity and $(B_{t})$ is the Brownian motion under the probability $\mathbb{P}$. The Cox-Ingersoll-Ross (CIR) model is one of the most popular and commonly used stochastic intensity in both academic research and practical applications. The process was first developed in \cite{CIR} to model the term structure of interest rates; It is set up as a single-good, continuous time economy with a single state variable. Multivariate versions are developed later by \cite{longs} and \cite{chen}. \\
 When we impose the condition $2\kappa\theta>\sigma^{2}$ then the intensity $\lambda$ is always positive, otherwise we can only guarantee that it is non-negative (with a positive probability to terminate to zero). In fact, when the fraud intensity approaches $0$ then the volatility $\sigma\sqrt{\lambda_{t}}$ approaches $0$ cancelling the effect of the randomness, so the intensity rate remains always non-negative. Figure \ref{fig1} shows the simulations of the stochastic fraud intensity following the CIR model with various parameters. All simulations generate dynamic non-negative intensities which tend to move around a long-run mean $\theta$. 

\begin{figure}[h]
\advance\leftskip-1cm
\includegraphics[width=0.6\textwidth,left]{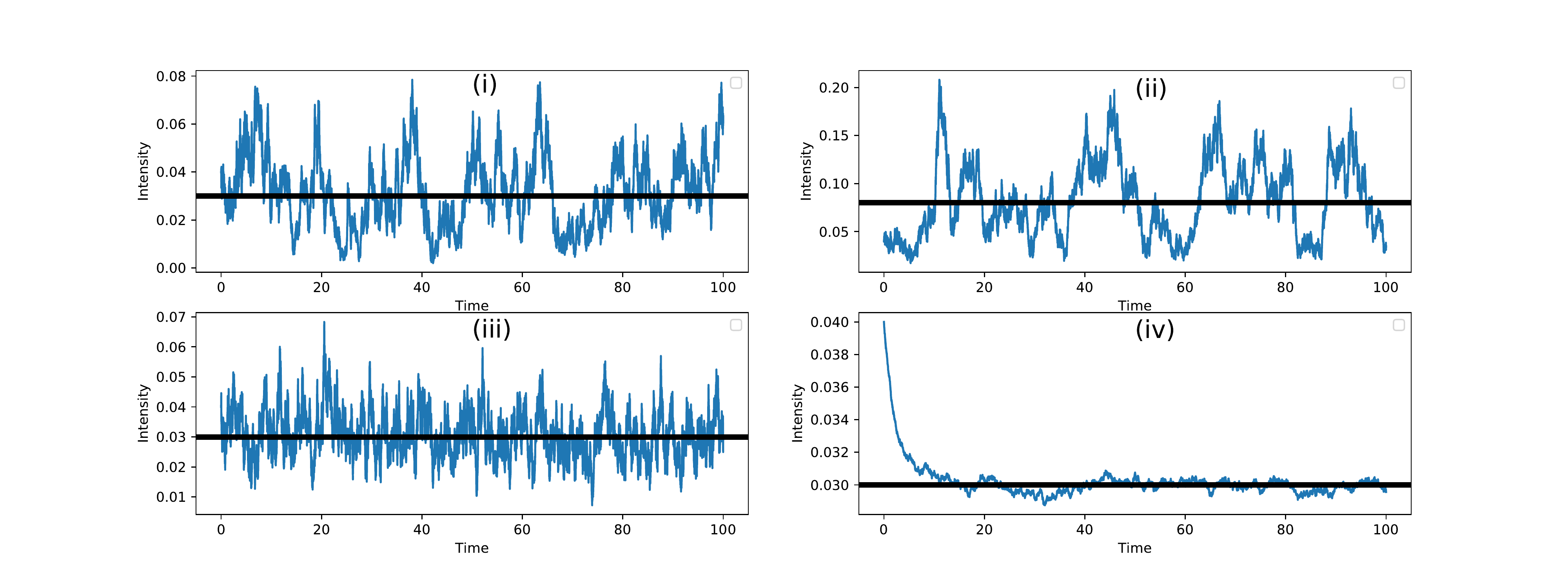}
\caption{Simulation of stochastic intensities with the CIR process. (i) $\lambda_{0}=0.04$, $\kappa=0.4$, $\theta=0.03$, $\sigma=0.1$. (ii) as for (i) but $\theta=0.08$. (iii) as for (i) but $\kappa=2$. (iv) as for (i) but $\sigma=0.002$.}
\label{fig1}
\end{figure}

In \cite{jaf}, it has been shown that
\begin{equation}
\lambda_{t}=e^{-\kappa t}\lambda_{0}+\theta(1-e^{-\kappa t})+\sigma e^{-\kappa t} \int^{t}_{0} e^{\kappa s}\sqrt{\lambda_{s}}dB_{s}
\label{cirt}
\end{equation}
with
\begin{equation}
E(\lambda_{t})=e^{-\kappa t}\lambda_{0}+\theta(1-e^{-\kappa t})
\label{mcirt}
\end{equation}
and
\begin{equation}
V(\lambda_{t})=\dfrac{\sigma^{2}}{\kappa}\lambda_{0}(e^{-\kappa t}-e^{-2\kappa t})+\dfrac{\theta \sigma^{2}}{2 \kappa}(1-e^{-\kappa t})^{2}
\label{vcirt}
\end{equation}
There is no general explicit solution for equation (\ref{cirt}). However, its calibration is critical for obtaining meaningful results. One of the easiest methods to implement it is to perform discretization of equation (\ref{cirt}) and then use available data for small time intervals, in order to be able to estimate the parameters. Let $\vartriangle_{t}=\frac{T}{N+1}$ and $t_{j}=j \cdot \vartriangle_{t}$ for  $j=0,...,N+1$. Equation (\ref{cirt}) becomes
\begin{equation}
\begin{aligned}
\lambda_{t_{i}}=e^{-\kappa \vartriangle_{t}}\lambda_{t_{i-1}}+ & \theta(1-e^{-\kappa \vartriangle_{t}})\\
& +\sigma e^{-\kappa t_{i}} \int^{t_{i}}_{t_{i-1}} e^{\kappa s}\sqrt{\lambda_{s}}dB_{s}
\end{aligned}
\label{mmcirt}
\end{equation}
From equation (\ref{mcirt}), $$E(\lambda_{t_{i}} \vert \lambda_{t_{i-1}})=\mu_{1}=e^{-\kappa \vartriangle_{t}}\lambda_{t_{i-1}}+\theta(1-e^{-\kappa \vartriangle_{t}})$$ 
and from equation (\ref{vcirt}) 
\begin{equation*}
\begin{aligned}
V(\lambda_{t_{i}} \vert \lambda_{t_{i-1}})=\mu_{2}= & \dfrac{\theta \sigma^{2}}{2 \kappa}(1-e^{-\kappa \vartriangle_{t}})^{2} \\
& +\dfrac{\sigma^{2}}{\kappa}(e^{-\kappa \vartriangle_{t}}-e^{-2\kappa \vartriangle_{t}})\lambda_{t_{i-1}}
\end{aligned}
\end{equation*}
equation (\ref{mmcirt}) is reduced to
\begin{equation}
\lambda_{t_{i}}=e^{-\kappa \vartriangle_{t}}\lambda_{t_{i-1}}+\theta(1-e^{-\kappa \vartriangle_{t}})+\epsilon_{t_{i}}
\label{ep}
\end{equation}
where $\epsilon_{t_{i}}=\sigma e^{-\kappa t_{i}} \int^{t_{i}}_{t_{i-1}} e^{\kappa s}\sqrt{\lambda_{s}}dB_{s}$ with $\epsilon_{t_{i}}$ which is an Ito integral with respect 
to the Brownian motion $(B_{t})$. Using the zero mean property $E(\epsilon_{t_{i}} \vert \epsilon_{t_{i-1}})=0$. From (\ref{ep})
$V(\epsilon_{t_{i}} \vert \epsilon_{t_{i-1}})=V(\lambda_{t_{i}} \vert \lambda_{t_{i-1}})=\mu_{2}$. Following \cite{CIR}, $\lambda_{t_{i}}$ given $\lambda_{t_{i-1}}$ is a non-central $\chi^{2}$ distribution with the first two moments $\mu_{1}$ and $\mu_{2}$. From \cite{Ball} and under the assumption of small time intervals, $\lambda_{t_{i}}$ given $\lambda_{t_{i-1}}$ can be reasonably approximated by a normal distribution with mean $\mu_{1}$ and variance $\mu_{2}$. Then,
$\lambda_{t_{i}}=e^{-\kappa \vartriangle_{t}}\lambda_{t_{i-1}}+\theta(1-e^{-\kappa \vartriangle_{t}})+\epsilon_{t_{i}}$ with $$\epsilon_{t_{i}} \sim \mathcal{N}(0,\dfrac{\theta \sigma^{2}}{2 \kappa}(1-e^{-\kappa \vartriangle_{t}})^{2}+\dfrac{\sigma^{2}}{\kappa}(e^{-\kappa \vartriangle_{t}}-e^{-2\kappa \vartriangle_{t}})\lambda_{t_{i-1}}).$$ Let us define $\alpha=\theta(1-e^{-\kappa \vartriangle_{t}})$, $\beta=e^{-\kappa \vartriangle_{t}}$ and $$\eta_{t_{i}}^{2}=\dfrac{\theta \sigma^{2}}{2 \kappa}(1-e^{-\kappa \vartriangle_{t}})^{2}+\dfrac{\sigma^{2}}{\kappa}(e^{-\kappa \vartriangle_{t}}-e^{-2\kappa \vartriangle_{t}})\lambda_{t_{i-1}}.$$ We can write
\begin{equation}
\lambda_{t_{i}}=\alpha + \beta \lambda_{t_{i-1}}+\epsilon_{t_{i}}
\label{sta1}
\end{equation}
with $\epsilon_{t_{i}} \sim \mathcal{N}(0,\eta_{t_{i}}^{2})$. According to the equation (\ref{sta1}), if we suppose that $V(\epsilon_{t})$ is constant, the process of $\{\lambda_{t_{i}}\}$ is a stationary AR(1) process. 

\subsection{Fraud detection in the context of the CIR intensity} 
\subsubsection{Prediction of fraud with the CIR intensity} 
We suppose that all the background information on the financial institution's transactions, except for the hours of fraud events fraud, is expressed by the filtration $\mathbb{G}=(\mathcal{G}_{t})_{t\geq 0}$. For example, $\mathcal{G}_{t}$ can be generated by a d-dimensional driving process $X_{t}$ which includes the information on transactions amounts, transaction dates, country of the receiving bank, client IDs, etc... Suppose further that there is a non-negative process $\lambda_{t}$ which is also adapted to $\mathbb{G}$ which plays the role of a stochastic intensity, generally correlated with the various components of the driving process $X_{t}$. Next assume that $\mathbb{H}=(\mathcal{H}_{t})_{t\geq 0}$ is the filtration generated by the fraud indicator process $1_{\{\tau \leqslant t\}}$. The full filtration for the model is obtained as $\mathbb{F}=\mathbb{G} \vee \mathbb{H}$ where $\mathbb{F}=(\mathcal{F}_{t})_{t\geq 0}$. Let $N=(N_{t})_{t\geq 0}$ where $N_{t}=\sum_{n\geqslant0} 1_{\{\tau_{n}\leqslant t\}}$, a counting process in the occurrence of fraud in a client's transactions. We say that $N=(N_{t})_{t\geq 0}$ is a doubly-stochastic Poisson process or a Cox process if, conditioned on the background information $\mathcal{G}_{t}$ available at time $t$, $N_{t}$ is an inhomogeneous Poisson process with a time-varying intensity $\lambda_{s}$, $0\leq s\leq t$. In other words, each realization of the process $\lambda_{t}$ determines the local jump probabilities for the process $N_{t}$. The intuition of the doubly-stochastic assumption is that $\mathcal{G}_{t}$ contains enough information to reveal the intensity $\lambda_{t}$, but not enough information to reveal the event times of the counting process $N$. That is why, the fraud time $\tau$ is a $\mathbb{F}$-stopping time but not a $\mathbb{G}$-stopping time. 
\begin{proposition}
Consider the filtration $\mathcal{F}_{t}$ that contains the information about the fraud events up to time $t$. Suppose that a new transaction is in progress at time $s$ ($s>t$). The probability of fraud occurring on the next transaction at time $s$ is given by 
\begin{equation}
P \left( N_{s}-N_{t}=1 \vert \mathcal{F}_{t} \right)=1-E\left(e^{-\int^{s}_{t}\lambda_{u} du} \vert \mathcal{F}_{t} \right)
\end{equation}
\label{bid}
\end{proposition}
\begin{proof}
Letting $A$ be the event $\left\{N_{s}-N_{t}\right\}$ of no fraud arrivals, the law of iterated expectations implies that, for $\tau>t$
\begin{align}
\begin{aligned}
P\left(N_{s}-\right.&\left.N_{t}=1 \vert \mathcal{F}_{t}\right)=1-P\left(N_{s}-N_{t}=0 \vert \mathcal{F}_{t}\right)\\
&=1-E\left(1_{A} \vert \mathcal{F}_{t}\right)\\
&=1-E\left(E\left(1_{A} \vert \mathcal{G}_{s} \vee  \mathcal{H}_{t}\right)\vert \mathcal{F}_{t}\right)\\
&=1-E\left(P\left(N_{s}-N_{t}=0 \vert \mathcal{G}_{s} \vee  \mathcal{H}_{t}\right)\vert \mathcal{F}_{t}\right)\\
&= 1-E\left(e^{-\int^{s}_{t}\lambda_{u} du} \vert \mathcal{F}_{t} \right)\\
\end{aligned}
\end{align}
\end{proof}
The last equation is derived by the fact that under the background information $\mathcal{G}_{s}$, $N_{s}$ is an inhomogeneous Poisson process.
\begin{proposition}
Suppose the intensity follows the stochastic  CIR process
\begin{equation*}
d\lambda_{t}= \kappa (\theta-\lambda_{t})dt+\sigma\sqrt{\lambda_{t}}dB_{t}
\end{equation*}
Under the assumptions of proposition (\ref{bid}), the probability that the fraud will occur on the next transaction at the time $s$ is given by 
\begin{equation}
P\left(N_{s}-N_{t}=1 \vert \mathcal{F}_{t}\right)=1-e^{A(s-t)-B(s-t)\lambda_{t}}
\label{fraud}
\end{equation}

where $\gamma=\sqrt{\kappa^{2}+2\sigma^{2}}$\\

\begin{equation*}
 B(s-t)=\frac{2 (e^{\gamma (s- t)}-1)}{2 \gamma+(\gamma +\kappa)(e^{\gamma (s-t)}-1)}
 \label{B}
 \end{equation*}
 
\begin{equation*}
 A(s-t)=\log \left[ \left\{\frac{2\gamma e^{(\kappa+\gamma)(\frac{(s-t)}{2})}}{2 \gamma+(\gamma+\kappa)(e^{\gamma (s-t)}-1)}\right\}^{ \frac{2\kappa   \theta}{\sigma^{2}}}\right] 
 \label{A}
\end{equation*}
\label{bida}
\end{proposition}

\begin{proof}
Using proposition (\ref{bid}) and following \cite{CIR} and \cite{CIRE}, we obtain (\ref{fraud}).
\end{proof}

Consequently, the prediction probability of fraud at time $s$ depends on the underlying parameters and on the dynamic intensity $\lambda_{t}$ with $s\geq t$. In the next section, we will focus on the form of relationship between the intensity and probability of fraud.

\subsubsection{Defining the measurement equation} 
Suppose we are now interested in the probability that no fraud will occur at the time $s$ given the filtration  $\left( \mathcal{F}_{t}\right)$  with $s>t$. We denote $Y_{s}$, the logarithm of this probability. From the proposition (\ref{bida}), $Y_{s}=\log \left[  P\left(N_{s}-N_{t}=0 \vert \mathcal{F}_{t}\right)\right] =A(s-t)-B(s-t)\lambda_{t}$. Let $a=A(s-t)$ and $b=-B(s-t)$; we have 
\begin{equation}
Y_{s}=a + b \lambda_{t}
\label{mas}
\end{equation}
The equation (\ref{mas}) shows an affine relationship between the logarithm of the probability of prediction for no fraud at time $s$ and the intensity at time $t$ with $s>t$. For simplicity and calibration reasons, we take $Y_{t}$ as proxy for $Y_{s}$ with $Y_{t}$ is the logarithm of the posterior probability for no fraud at time $t$. The main reason for choosing $Y_{t}$ instead of $Y_{s}$ is that $Y_{t}$ is known at time $t$ and therefore it will be useful later in the filtering methods. So, we introduce noises in the equation (\ref{mas}) to take into account the differences between $Y_{t}$ and $Y_{s}$. We assume that these noises are Gaussian white noises. Therefore, equation (\ref{mas}) is written 
\begin{equation}
Y_{t}=a + b \lambda_{t}+\mu_{t}
\label{mas1}
\end{equation}
with $\mu_{t}\sim \mathcal{N}(0,w_{t}^{2})$. Although the equation (\ref{mas1}) is affine in the state $\lambda_{t}$, the functions $a$ and $b$ are non-linear functions of the underlying parameters. Also for $s<t$, we always have $b<0$; this implies a negative relationship between the likelihood of no fraud and the intensity of fraud. However, despite the Gaussian assumption of $\epsilon_{t}$ in the autoregressive equation (\ref{sta1}), the maximum likelihood estimation of the intensity's parameters is no longer feasible because the intensity $\lambda_{t}$ is an unobserved variable and the probability density function is not available in a closed form. On the other hand, taking the equation (\ref{mas1}), it would be difficult to estimate the parameters $a$ et $b$ by the likelihood estimation for the same reason. Thus, filtering methods can be used to track the intensities based on the observed probabilities $Y_{t}$. The Kalman filter is proposed here to capture the dynamic intensities and to estimate the various parameters. In this context, two equations are required; the measurement equation (\ref{mas1}) that concerns the observed probability $Y_{t}$ and the state equation (\ref{sta1}) for the unknown intensity.

\subsection{Estimation of of the stochastic intensity} 
\subsubsection{Kalman filter in the estimation of the stochastic intensity} 
Now that the model in (\ref{lcir}) has been put in state space form conducting  to equations (\ref{sta1}) and (\ref{mas1}), the Kalman filter can be used to obtain information about the unobserved intensity $\lambda_{t}$ using the logarithm of observed probability for no fraud, $Y_{t}$ for $t=0,..,n$. Let's recall the measurement and state equations. Measurement equation: $Y_{t}=a + b \lambda_{t}+\mu_{t}$, with $\mu_{t}\sim \mathcal{N}(0,w_{t}^{2})$. State equation:
$\lambda_{t}=\alpha + \beta \lambda_{t-1}+\epsilon_{t}$, with $\epsilon_{t} \sim \mathcal{N}(0,\eta_{t}^{2}).$ where $a$, $b$, $\alpha$, $\beta$ and $\epsilon_{t}$ are functions of the unknown parameters $(\kappa,\theta,\sigma)$ of the model. The Kalman filter is actually a recursive algorithm for calculating estimates of unobserved state variables based on observations that depend on these state variables. It was first published in \cite{kalman} and it is used in areas as aeronautics, signal processing, and futures trading. A detailed explanation of the Kalman filter can be found in \cite{Harvey}, \cite{luth}, \cite{maybeck}, \cite{jaz} and \cite{hee}. The principle of the Kalman filter is to use a time series of observable data to estimate the values of state variables. This technique is useful when there is a linear dependency of the observable data on the state variables. In our case, we have this linearity relation between the probability of no fraud and the fraud intensity. The algorithm first forms an optimal predictor of the unobserved state variable vector given its previous estimated value. This prediction is obtained by using the distribution of the unobserved state variables, conditional on the previous estimated values. These estimates for the unobserved state variables are then updated using the information provided by the observed variables. Although the Kalman filter relies on the  normality assumption of the measurement error and initial state vector, one can calculate the likelihood function by decomposing the prediction error. Let 
$v_{t}$ be the variance of $\lambda_{t}$, $\lambda_{t-1 \vert t-1}$ an unbiased estimation of $\lambda_{t-1}$ at time $t-1$ and $v_{t-1 \vert t-1}$ the variance of $\lambda_{t-1 \vert t-1}$. The initial state $\lambda_{0}$ at time $0$ is a random variable which is not correlated with both the system and the measurement noise processes. 
A time $0$, we must have a preliminary value of $\lambda_{0 \vert 0}$ and $v_{0 \vert 0}$. As these values are unknown, a common way is to put a null value to $\lambda_{0 \vert 0}$ and a high value to $v_{0 \vert 0}$ in order to take into account the uncertainty linked to the estimate of $v_{0 \vert 0}$.  Let us give the three steps of the procedure followed by the Kalman filter: forecasting, updating and estimation of the parameters. First, we make the following forecasts:

\begin{enumerate}
\item $\lambda_{t \vert t-1}=E_{t-1}(\lambda_{t})$, that is the forecast of $\lambda_{t}$ conditional to the information set at time $(t-1)$. 
\begin{equation}
\lambda_{t \vert t-1}=\alpha + \beta \lambda_{t-1 \vert t-1}
 \end{equation}
$\lambda_{t \vert t-1}$ is an unbiased conditional estimation of $\lambda_{t}$. In fact, It is straightforward to check that  $E(\lambda_{t \vert t-1}-\lambda_{t})=0$.
\item  $v_{t \vert t-1}$ as the variance of $\lambda_{t \vert t-1}$, which is $v_{t \vert t-1}=E[(\lambda_{t \vert t-1}-\lambda_{t})^{2}]$. 
\begin{equation}
v_{t \vert t-1}=\beta^{2} v_{t-1 \vert t-1}+ \eta_{t}^{2}
\end{equation}
\end{enumerate}
The two forecasts $\lambda_{t \vert t-1}$ and $v_{t \vert t-1}$ will be used in the next step to update $\lambda_{t}$ and its variance. The second step is the update. At time $t$, we have a new observation of $Y$, i.e. $Y_{t}$. We can thus compute the prediction error $e_{t}$:
\begin{equation}
e_{t}=Y_{t}-a-b\lambda_{t \vert t-1}
\label{err1}
 \end{equation} 
 The variance of $e_{t}$, denoted by $\psi_{t}$ is given by :
\begin{equation}
\psi_{t}=w_{t}^{2}+b^{2}v_{t \vert t-1}
 \end{equation}  
 We use $e_{t}$ and $\psi_{t}$ to update $\lambda_{t \vert t}$ and its variance $v_{t \vert t}$ as follows 
\begin{equation}
\lambda_{t \vert t}=\lambda_{t \vert t-1}+K_{t}e_{t}
\label{unb}
 \end{equation} 
\begin{equation}
v_{t \vert t}=(1-K_{t}b)v_{t \vert t-1}
\label{unb1}
\end{equation}
with $K_{t}$ is Kalman gain defined as $K_{t}=\frac{b v_{t \vert t-1}}{\psi_{t}}$. The Kalman gain $K_{t}$ is  the  most  crucial  parameter  of  the  filter. This determines how easily the filter will adapt to all  possible new  conditions. In (\ref{unb}), $K_{t}$ guarantees that $\lambda_{t \vert t}$ will be an unbiaised estimator of $\lambda_{t}$. In (\ref{unb1}), it minimizes the variance $v_{t \vert t}$. Thus, $\lambda_{t \vert t}$ is a conditionally unbiased and efficient estimator. The Kalman filter is therefore optimal because it is the best estimator in the class of linear estimators. For more details on the Kalman gain derivation, see \cite{Hamilton} and \cite{Welch}. The third step concerns the estimation of the parameters. In our study, $4$ parameters have to be estimated: $\kappa$, $\theta$, $\sigma$ and the variances of the measurement error at each time step, $w_{t}$. From (\ref{err1}), the prediction error $e_{t}$ follows the normal distribution with mean $0$ and variance $\psi_{t}$. Based on the Gaussian distribution of $e_{t}$, we use the maximum likelihood method. The log-likelihood function can be written as follows:
\begin{equation}
l=-\frac{1}{2} \sum_{t} \text{log}(\psi_{t})-\frac{1}{2} \sum_{t} \frac{e_{t}^{2}}{\psi_{t}}
 \end{equation} 
To complete the procedure, we go to time $(t+1)$ and repeat the three-step procedure up to $n$. As discussed in \cite{Duan}, when the state space model is Gaussian, the Kalman filter provides an optimal solution to predict, update and evaluate the likelihood function. When the state-space model is non-Gaussian, the Kalman filter can still be applied to obtain approximate first and second moments of the model and the resulting filter is almost optimal. The use of this quasi-optimal filter gives an approximate quasi-likelihood function with which the estimation of the parameters can be  performed. So, our fraud detection approach is an unsupervised approach in the sense that the estimation of the dynamic intensities does not require the labels but the fraud probabilities observed on the transactions. This approach will be useful for the detection of fraudulent transaction for which the main constraint is the lack of labelling dataset.

\subsubsection{Issues with negative estimated values for fraud intensity} 

From the equation (\ref{lcir}), the intensity follows a non-central $\chi^{2}$ distribution and this guarantees that the intensity is always non-negative. However, the intensity is unobservable variable and in order to estimate its values, the approach by the Kalman filter is proposed. As noted in the previous section the Kalman filter uses the quasi maximum likelihood to estimate the intensity, since the true distribution of the intensity is not Gaussian. 
Therefore, there is a non-zero probability to obtain negative values for the intensity during the calibration process.
To deal with the possible negative values of the intensity, the following steps are proposed.

\paragraph{Step 1: Intensities Shift}  
This step consists in translating the intensity values obtained by the Kalman filter ($\lambda_{t \vert t}$) to positive values ($S_{t \vert t}$) to eliminate negative/near-zeros values. The following transformation is proposed

\begin{equation}
S_{t \vert t}= \lambda_{t \vert t} + \alpha, \,\,\, t \in [0,n] 
\end{equation}
where $\alpha$ is a deterministic positive quantity. From the above translation, $dS_{t}=d \lambda_{t}$ for any time $t$. There are many values that could be assigned to $\alpha$, but in our study the most appropriate choice is the 99th percentile of the empirical distribution of the intensity. The Stochastic Differential Equation (SDE) of $S_{t}$ becomes:
\begin{eqnarray*}
dS_{t}&=& d\lambda_{t}\\
&=& \kappa (\theta-\lambda_{t})dt+\sigma\sqrt{\lambda_{t}}dB_{t}\\
&=& \kappa (\theta-(S_{t}-\alpha))dt+\sigma \sqrt{S_{t}-\alpha} dB_{t}\\
&=& \kappa (\theta+\alpha-S_{t}) dt+\sigma \sqrt{S_{t}(1-\frac{\alpha}{S_{t}})} dB_{t}\\
&=& \kappa (\theta+\alpha-S_{t}) dt+\sigma \sqrt{1-\frac{\alpha}{S_{t}}} \sqrt{S_{t}} dB_{t}
\end{eqnarray*}
\begin{equation}
dS_{t}= \kappa (\theta^{*}-S_{t}) dt+\sigma_{t}^{*} \sqrt{S_{t}} dB_{t}
\label{nint}
\end{equation}

with $\theta^{*}=\theta+\alpha$ and $\sigma_{t}^{*}=\sigma \sqrt{1-\frac{\alpha}{S_{t}}}$. $S_{t}$ follows an extended CIR with stochastic $\sigma_{t}^{*}$. $S_{t}$ is a mean reverting process with $\kappa$ being the rate of mean reverting, $\theta^{*}$ the long run average and $\sigma_{t}^{*}$ the volatility. If $S_{t}$ approaches $\alpha$, $\sigma_{t}^{*}=\sigma \sqrt{1-\frac{\alpha}{S_{t}}}$ approaches $0$ cancelling the effect of randomness, so $S_{t}\geq \alpha$. 
\paragraph{Step 2: Updating parameters} The SDE in (\ref{nint}) does not lead to the analytical expression of the proposition (\ref{bida}) because $\sigma_{t}^{*}$ is stochastic but not time-dependent; see \cite{boyle}. In order to apply (\ref{fraud}) to predict the fraud occurrence with the new intensity $S_{t}$, the SDE of $S_{t}$ is modified as follows 
\begin{equation}
dS_{t}= \kappa (\theta^{*}-S_{t}) dt+\sigma^{*} \sqrt{S_{t}} dB_{t}
\label{nnint}
\end{equation}
with $\sigma^{*}=E(\sigma_{t}^{*})$. In this context, the parameters $\kappa$, $\theta$ and $\sigma^{*}$ can be updated by Ordinary Least Square (OLS). The discretised form of equation (\ref{nnint}) is given by
\begin{equation}
S_{t+\Delta t}-S_{t}= \kappa (\theta^{*}-S_{t}) \Delta t+\sigma^{*} \sqrt{S_{t}} \xi_{t}
\label{ds}
\end{equation}
where $\xi_{t}$ is a Gaussian white noise with $E(\xi_{t})=0$ and $V(\xi_{t})=\Delta t$. For performing OLS, we transform (\ref{ds}) by
\begin{equation}
\frac{S_{t+\Delta t}-S_{t}}{\sqrt{S_{t}}}= \frac{\kappa \theta^{*}\Delta t}{\sqrt{S_{t}}}-\kappa \sqrt{S_{t}} \Delta t +\sigma^{*} \xi_{t}
\end{equation}
Then, the drift parameters $\kappa$ and $\theta^{*}$ are found by minimizing the OLS objective function 
\begin{equation}
\sum_{i=1}^{n-1} \left( \frac{S_{t_{i+1}}-S_{t_{i}}}{\sqrt{S_{t_{i}}}}- \frac{\kappa \theta^{*}\Delta t}{\sqrt{S_{t_{i}}}}+\kappa \sqrt{S_{t_{i}}} \Delta t\right) ^{2} 
\end{equation}
The diffusion parameter estimate $\hat{\sigma^{*}}$ is found by dividing the standard deviation of residuals by $\sqrt{\Delta t}$. So, in the context of negative values for $\lambda_{t}$ the updated parameters $\kappa$, $\theta^{*}$ and $\hat{\sigma^{*}}$ for the new intensity $S_{t}$ are finally used in proposition (\ref{bida}) for fraud prediction.

\begin{table}[t]
\caption{Summary Statistics of Risk-Scores by transaction and of fraud proportion by client in the full dataset. The clients with no fraud events and the clients with $100\%$ of fraud proportion are removed from this full dataset.}
\label{table:1}
\vskip 0.15in
\begin{center}
\begin{small}
\begin{sc}
\resizebox{\columnwidth}{!}{
\begin{tabular}{lccccr}
\toprule
 &Min & Max & Mean & Median & St dev\\
\midrule
Risk-Score & 0 & 0.689 & 0.02 & 0.007 & 0.04\\ 
 Fraud Proportion & 0.0004 & 0.99 & 0.22 & 0.09 & 0.26\\
\bottomrule
\end{tabular}
}
\end{sc}
\end{small}
\end{center}
\vskip -0.1in
\end{table}

\section{Datasets}
The data provided by NetGuardians is a simulated banking transactions dataset created by NetGuardians from anonymized real-world banking datasets. It covers a period of 2 years and contains a total of more than 15 millions transactions made by more than 120'000 clients. The dataset includes a total of $49$ features such as the transaction dates, transaction amounts, transaction senders IDs, the account numbers of transaction recipients, bank countries receiving transactions, etc... It is important to mention that there is no fraudulent labelling in the dataset. 

\begin{figure}[H]
\begin{center}
\advance\leftskip-1.3cm
\includegraphics[width=0.55\textwidth,left]{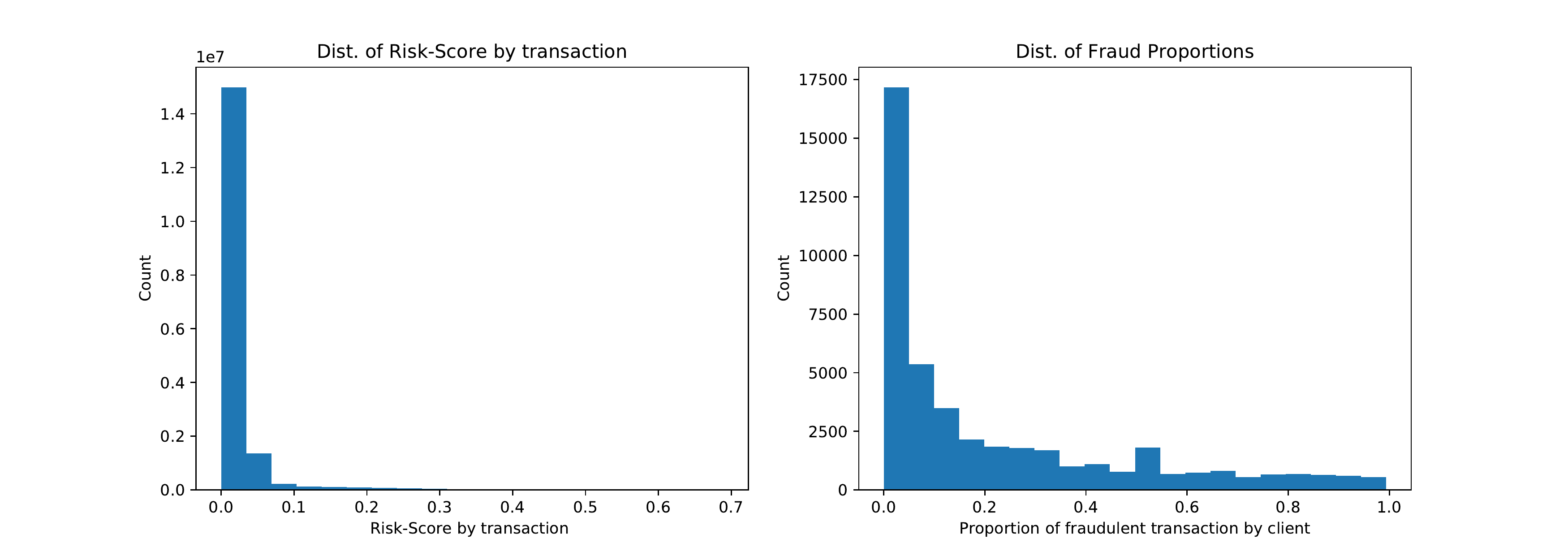}
\end{center}
\caption{Histogram of fraud Risk-Score and of fraud proportion by client in the full dataset. The clients with no fraud events and the clients with $100\%$ of fraud proportion are removed from this full dataset. }
\label{fig2}
\end{figure}

\begin{table}[t]
\caption{Repartition of the number of the clients and the number of distinct transactions in the 7 subsets.}
\label{table:2}
\vskip 0.15in
\begin{center}
\begin{small}
\begin{sc}
\resizebox{\columnwidth}{!}{
\begin{tabular}{lcccr}
\toprule
 Groups & Nb of clients & Nb of dist trans\\
\midrule
$P\leq 0.004$ & 1391 & 627555\\ 
 $0.004<P \leq 0.006$ & 1557 & 431456\\
 $0.006<P \leq 0.01$ & 2487 & 609649\\
 $0.01<P \leq 0.05$ & 11735 & 2090126\\
 $0.05<P \leq 0.08$ & 3802 & 531486\\ 
 $0.08<P \leq 0.1$ & 1785 & 231305\\   
 $0.1<P \leq 0.15$ & 2312 & 401052\\ 
 Total & 26069 & 4922629\\\bottomrule
\end{tabular}
}
\end{sc}
\end{small}
\end{center}
\vskip -0.1in
\end{table}
The model with the Kalman filter is unsupervised learning in the sense that no label is required to detect the fraud event. Instead, information about the likelihood of fraud for each transaction is required to define the measurement equation. Many machine learning or statistical methods such as the dimensionality reduction method, logistic regression, Z-Score, etc...can be used to estimate the fraud probability. However, the performance of the model will strongly depend on the accuracy of such method in estimating the fraud probability. In this study, for the sake of simplicity we focus on the fraud Risk-Score provided by NetGuardians to define the measurement equation. Risk-Score fraud is a metric that gives an estimate of the fraud proportion for each transaction based on the recent information. In addition for reasons of confidentiality the methodology for calculating the Risk-Score will not be mentioned.
\begin{table}[H]
\caption{Process of fraud prediction in the test set: The intensity is updated by the Kalman filter and the fraud is predicted on the next transaction.\\}
\centering
\resizebox{\columnwidth}{!}{
\begin{tabular}{||c||} 
 \hline
   \text{At time $t$}\\ [0.5ex] 
 \hline
Given starting $\lambda_{t \vert t}$ at time $t$\\ 
$P\left(N_{t+1}-N_{t}=1 \vert \mathcal{F}_{t}\right)=\text{Fraud prediction for time $t+1$}$\\[1ex]
 \hline \hline
 \text{At time $t+1$}\\ [0.5ex]
 \hline
Risk-Score at time $t+1$ is provided\\ 
$\lambda_{t+1 \vert t+1}$ is updated by the Kalman Filter\\ 
$P\left(N_{t+2}-N_{t+1}=1 \vert \mathcal{F}_{t+1}\right)=$\\
$\text{Fraud prediction for time $t+2$}$\\[1ex]
\hline \hline
\text{At time $t+2$}\\ [0.5ex]
\hline
Repeat the process as at time $t$ and so on ...\\
\hline 
\end{tabular}
}
\label{table:3}
\end{table}
To complete our study, we also need to generate artificial fraud labels in our dataset. The main reason is that we expect to compare the performance of the Kalman filter model with other intensity-based models such as Homogeneous and Non-homogeneous Poisson process which are supervised methods and have been investigated in \cite{regis}. There are several possibilities to create artificial labels according to the specified criterion. In our study, artificial labelling is based on the following criterion: transactions for which the banks receiving the money are located outside Switzerland are flagged as fraudulent. The main reason of using this criterion is that the provided labels are very correlated with the fraud Risk-Score. The point biserial correlation between the artificial labels and the Risk-Score is around $0.80$. The proportion of fraud which is the number of fraudulent transactions over the total number of transactions is calculated for each client. According to the labelling methodology, we found some clients with a fraud rate of $100\%$. This concerns the clients for whom the institutions receiving the money are all located outside of Switzerland. To be realistic, we remove these clients from our analysis. Also, the clients with no fraud events in the full dataset are removed because the datasets of these clients contain only one class and the classification problem is not defined. Table \ref{table:1} and and figure \ref{fig2} show the descriptive statistics of the fraud Risk-Scores and of the client's fraud proportions in the cleaned dataset. We remark that the two distributions are skewed and the dataset is unbalanced since most of the clients have small proportion of fraud. The mean and the median of the fraud proportions are $22\%$ and $9\%$ respectively. 

However, it is important to note that with the labelling criteria, the above distribution of the fraud proportions is not representative of the true fraud distribution because in practice the majority of fraud proportions are less than $1\%$. To investigate our analysis in a general framework of imbalanced dataset, we focussed on clients with a fraud proportion below $15\%$ which leads to a sample of $26,069$ clients with $4,922,629$ separate transactions. Next, we divide this sample in seven subsets containing various fraud profiles. The first subset regroups the clients for proportion of fraud less than $0.4\%$. The second subset concerns the clients with proportion between $0.4\%$ and $0.6\%$. The third subset concerns the clients with proportion between $0.6\%$ and $1\%$. The fourth subset concerns the clients with proportion between $1\%$ and $5\%$. The fifth subset concerns the clients with proportion between $5\%$ and $8\%$. The sixth subset concerns the clients with proportion between $8\%$ and $10\%$. The last subset concerns the clients with proportion between $10\%$ and $15\%$.

\begin{table*}[t]
\caption{Results for AUC medians for the various models in each group.}
\label{table:4}
\vskip 0.15in
\begin{center}
\begin{small}
\begin{sc}
\begin{tabular}{lccccccccr}
\toprule
 Models   & Group $1$ & Group $2$ & Group $3$ & Group $4$ & Group $5$ & Group $6$ & Group $7$\\
\midrule
HomoPoisson & 0.5 & 0.5 & 0.5 & 0.67 & 0.68 & 0.67 & 0.68\\ 
 LinearPoisson & 0.5 & 0.5 & 0.51 & 0.72 & 0.69 & 0.69 & 0.7\\ 
 QuadraticPoisson & 0.5 & 0.5 & 0.5 & 0.7 & 0.69 & 0.68 & 0.7\\
 NaiveApproach & 0.5 & 0.5 & 0.5 & 0.5 & 0.5 & 0.5 & 0.5\\
 ScoreApproach & 0.64 & 0.6 & 0.6 & 0.56 & 0.56 & 0.55 & 0.5\\
 KFApproach & 0.82 & 0.8 & 0.76 & 0.67 & 0.59 & 0.57 & 0.57\\
\bottomrule
\end{tabular}
\end{sc}
\end{small}
\end{center}
\vskip -0.1in
\end{table*}

\begin{table*}[t]
\caption{Results for A/P medians for the various models in each group.}
\label{table:5}
\vskip 0.15in
\begin{center}
\begin{small}
\begin{sc}
\begin{tabular}{lccccccccr}
\toprule
 Models   & Group $1$ & Group $2$ & Group $3$ & Group $4$ & Group $5$ & Group $6$ & Group $7$\\
\midrule
 HomoPoisson & 0.01 & 0.03 & 0.07 & 0.46 & 0.46 & 0.49 & 0.50\\ 
 LinearPoisson & 0.01 & 0.03 & 0.37 & 0.53 & 0.5 & 0.5 & 0.54\\ 
 QuadraticPoisson & 0.01 & 0.03 & 0.21 & 0.52 & 0.48 & 0.5 & 0.53\\
 NaiveApproach & 0.01 & 0.02 & 0.02 & 0.05 & 0.09 & 0.12 & 0.14\\
 ScoreApproach & 0.03 & 0.04 & 0.06 & 0.12 & 0.19 & 0.22 & 0.24\\
 KFApproach & 0.07 & 0.09 & 0.09 & 0.12 & 0.17 & 0.19 & 0.20\\
\bottomrule
\end{tabular}
\end{sc}
\end{small}
\end{center}
\vskip -0.1in
\end{table*}

The boundaries of the subsets are chosen to ensure a minimum number of 1000 clients in each subset. Table \ref{table:2} shows the distribution of the number of clients and the number of transactions in each group. Among the $7$ subsets, the first group contains the small number of clients and the fourth subset contains the large volume of clients. In figure \ref{fig3} the Boxplots for the proportion of fraud in each group are represented. In each subset, we select randomly a fixed number $N$ of clients and we train and test our model on the transactions for each client. $N=1391$ which represents the number of clients in the first group (the smallest group). The training set represents the first $80\%$ in chronological order of transactions for each client where the intensity parameters are estimated. The test set represents the last $20\%$ and the fraud events are predicted with the estimated parameters. We compare our model to other intensity-based model such as the Homogeneous and the Inhomogeneous Poisson process. For the Inhomogeneous Poisson process, we focussed on $\lambda(t)=a+bt$ and 
$\lambda(t)=a+bt+ct^{2}$ as in \cite{regis}. We consider two other models; the baseline model and the Risk-Score model. The baseline model consists of calculating the proportion of fraud in the training set and using this probability to predict fraud in the test set. The Risk-Score model consists in using the Risk-Score of the transaction at time $t$ as the fraud prediction on the next transaction of the client. All these models are compared to our Kalman filter model. Finally, the predictive performance is summarized in each subset using two performance measures: the ROC AUC and the Average Precision (AP) score.

\section{Results}
Our model is noted by \text{KFApproach}. In the training set for each client, the equations (\ref{sta1}) and (\ref{mas1}) are estimated by the Kalman filtering (\text{KF}) process as described above. We set the starting value of the intensity $\lambda_{0 \vert 0}=0$ and the variance $v_{0 \vert 0} =10$. As explained above a high value is given to the variance in order to take account the uncertainty in the estimation of the starting value $\lambda_{0 \vert 0}$. 

\begin{figure}[H]
\begin{center}
\advance\leftskip-1.5cm
\includegraphics[width=0.57\textwidth,left]{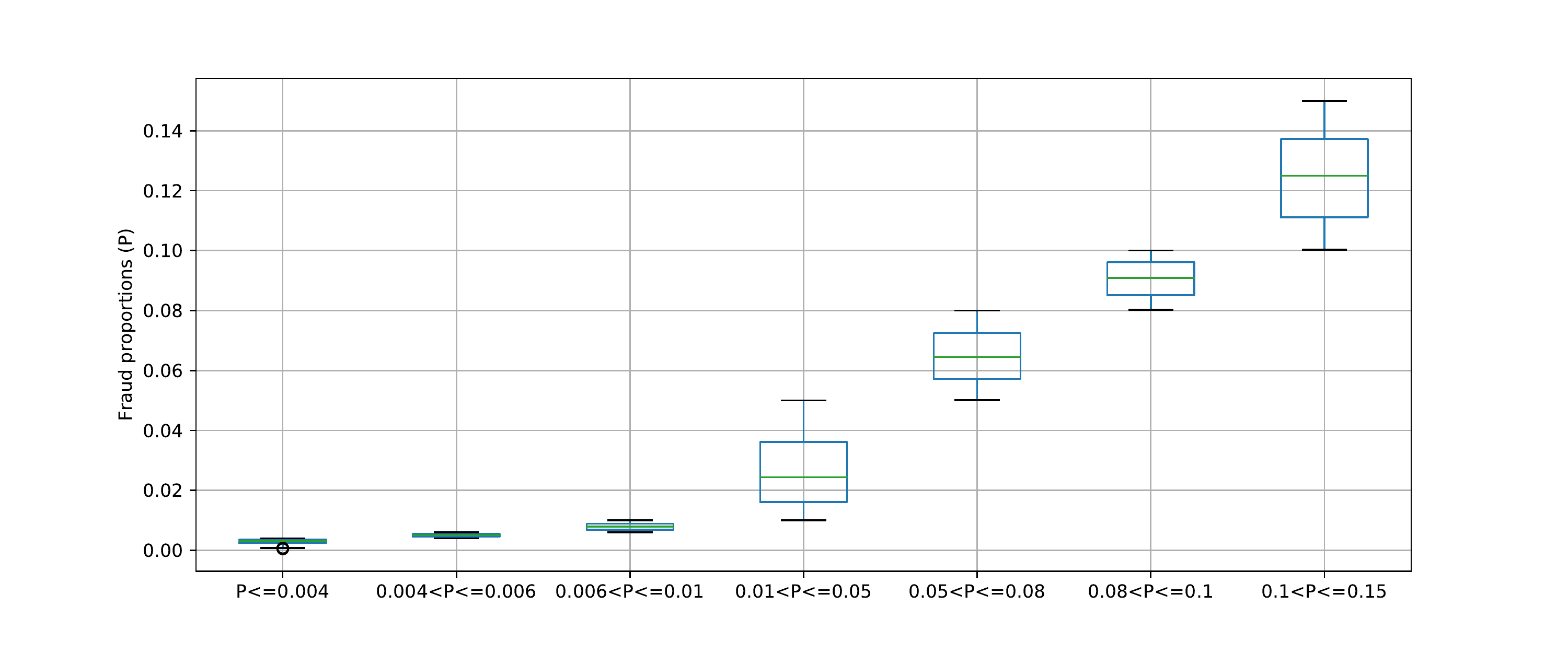}
\end{center}
\caption{Boxplots for the proportions of fraud in the seven subsets. Clients with the fraud proportion $P\leq 15\%$ are grouped in seven subsets representing different fraud profiles.}
\label{fig3}
\end{figure}

Thus in the training set, we generate with the \text{KF} the dynamic intensities related to the transactions of the client. In the test set, the prediction process described in table \ref{table:3} works as follows: At time $t$, given the updated intensity on transaction $\lambda_{t \vert t}$, the equation (\ref{fraud}) in proposition (\ref{bida}) is used to predict the fraud probability on the client's next transaction at time $t+1$. At time $t+1$, using the information provided by the Risk-Score on the transaction, the intensity at time $t+1$, $\lambda_{t+1 \vert t+1}$ is updated by the \text{KF}. With $\lambda_{t+1 \vert t+1}$, the fraud probability on transaction at time $t+2$ is predicted and the process is repeated until the last transaction in the test set

There is a total of $5$ models to compare with the Kalman filter model:
\begin{enumerate} 
\item The first model is the homogeneous Poisson process ($\lambda(t)=\lambda$). The constant intensity $\lambda$ is estimated in the training set. It is used to predict the fraud event in the whole test set. We designate this model by \text{HomoPoisson}.
\item The second model is the non-homogeneous Poisson process whose intensity is a linear function of time ($\lambda(t)=a+bt$). The intensity parameters are estimated in the training set and are used for the prediction of fraud in the whole test set. It is noted \text{LinearPoisson}.
\item The third model is the non-homogeneous Poisson process whose intensity is a quadratic function of time ($\lambda(t)=a+bt+ct^{2}$). The procedure is the same as in \text{LinearStatic}. We designate this model \text{QuadraticPoisson}
\item  the fourth model is the baseline model which consists of estimating the probability of fraud in the training set and using the same probability for the prediction in the test set. Thus, the predicting probabilities are the same for all transactions in the test set. This is equivalent to a random classifier because the model does not have the capacity to discriminate between an authentic transaction and a fraudulent transaction. We designate this model \text{NaiveApproach} 
\item The fifth model is based on the Risk-Score of the transactions and consists of predicting the Risk-Score using a Random Walk process. It supposes that the Risk-Scores follow the following process
\begin{eqnarray}
X_{t}=X_{t-1}+u_{t}
\end{eqnarray}
where $X_{t}$ is the Risk-Score of transaction at time $t$ and $u$ is a white noise. In this context, $X_{t}$ is not stationary and the best prediction of the fraud proportion on transaction at time $t+1$ is the fraud proportion at time $t$. This approach is indicated by \text{ScoreApproach}
\end{enumerate} 

\begin{figure}[H]
\begin{center}
\advance\leftskip-1.5cm
\includegraphics[width=0.57\textwidth,left]{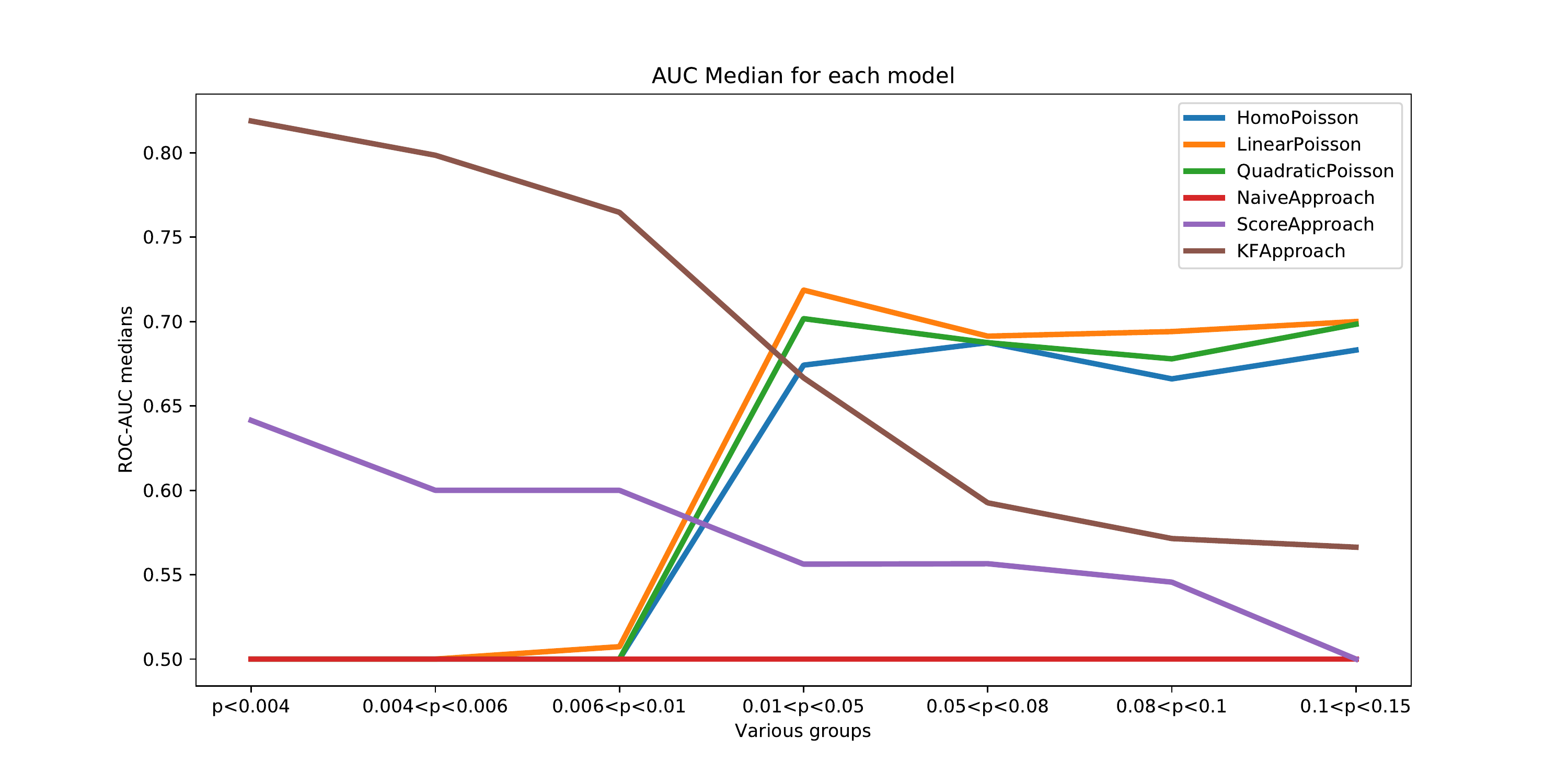}
\end{center}
\caption{Performance for AUC Medians for various models in each group. The plot shows the performance for each model with the degree of imbalanced dataset.}
\label{fig4}
\end{figure}

Below, we present the results of the performance for the $6$ models based on the predicting probabilities and the artificial labels in the test set. For comparison reasons, in each group the $1391$ performances calculated are summarized using the median. Table \ref{table:4} shows the results for the AUC (Area Under The Curve)-ROC (Receiver Operating Characteristics) for the different models in each group. The results show that \text{KFApproach} performs better than the other $5$  models in the group $1$ to group $3$ that is when the probability of fraud is less than $1\%$. It is followed by \text{ScoreApproach}. When the fraud probability is greater than $1\%$, all Poisson models outperform the \text{KFApproach} and \text{ScoreApproach}. We also remark that \text{NaiveApproach} works less well than the other models. It is important to note that \text{KFApproach} outperforms \text{ScoreApproach} in all groups; this can be attributed to the fact that Kalman filter combines information on the Risk-Scores of the transactions with additional information given by the instantaneous fraud rate for the client. In fact, the Risk-Score information is described by the measurement equation and the information on the instantaneous rate of fraud is described by the state equation.

Figure \ref{fig4} plots the AUC medians for the different models in each group. As mentioned above the Kalman filter model followed by  the \text{ScoreApproach} outperforms significantly the Poisson models in higher imbalanced dataset ($P<=1\%$) and this performance decreases when the probability of fraud for groups increases. We have the opposite effect for the Poisson models in the sense that their performance increases from $P>=1\%$ and becomes relatively stable up to $P<=15\%$. 

This concludes that the stochastic approach for the intensity is more adapted to the fraud prediction in high imbalanced dataset. Finally among the three Poisson models, \text{LinearStatic} is the best one followed by \text{QuadraticPoisson}; for more details, see \cite{regis}. We complete our analysis by focussing on the Precision-Recall performance. The Average Precision (A/P) is calculated which is an estimate of the area under the precision-recall curve and their results are summarized in the table \ref{table:5}. When $P<=0.6\%$, we notice that \text{KFApproach} followed by the \text{ScoreApproach} outperforms the Poisson models and the baseline approach. When $P>0.6\%$, all Poisson models except the \text{HomoPoisson} in the group $3$ perform better than \text{KFApproach}. The baseline approach still works less well than the other models. Figure \ref{fig5} shows the evolution of the A/P median with the probability of fraud and we observe that all models tend to increase with the degree of balanced dataset. It is important to note that with the A/P, \text{ScoreApproach} outperforms \text{KFApproach} when $P>5\%$. Finally, we conclude the important results: 1. \text{KFApproach} is a mixing approach combining the dynamic intensities with Risk-Scores. The ROC-AUC shows that \text{KFApproach} always outperforms the \text{ScoreApproach}; this shows that the prediction of the fraud probability by the Kalman filter is better than the prediction of the Random Walk process on the Risk-Scores. 2. \text{KFApproach} followed by the \text{ScoreApproach} works better than the other models in high imbalanced dataset. The analyzes on ROC-AUC and the  A/P confirm this result when $P<=1\%$ and $P<=0.6\%$ respectively. In fact, in a very imbalanced dataset, there is less fraud information and the intensity-based approach only is not enough for the prediction of fraud events. \text{KFApproach} uses additional information on the Risk-Score of the transactions and this explains why it outperforms the rest of the  models. So, the contribution of the Risk-Scores to the \text{KFApproach} in high imbalanced dataset is more significant. Therefore, \text{KFApproach} would be an interesting approach for detecting fraud in high imbalanced dataset. 3. Analysis of the AUCs shows that the performance of \text{KFApproach} as well as \text{ScoreApproach} decreases when the fraud probability of the dataset increases. On the other hand, A/P shows that the shapes of  the two approaches tend to be tilted upwards. 4. Analysis of the AUCs and the A/P shows that the Poisson models perform better than \text{KFApproach} and  \text{ScoreApproach} in more balanced dataset because more information of fraud events are available for estimating only the intensity. The A/P shows that \text{ScoreApproach} outperforms \text{KFApproach} when $P>5\%$ which is contrary to the analysis of the ROC-AUC. Finally, all models perform better than the baseline model.

\begin{figure}[H]
\begin{center}
\includegraphics[width=3.5in]{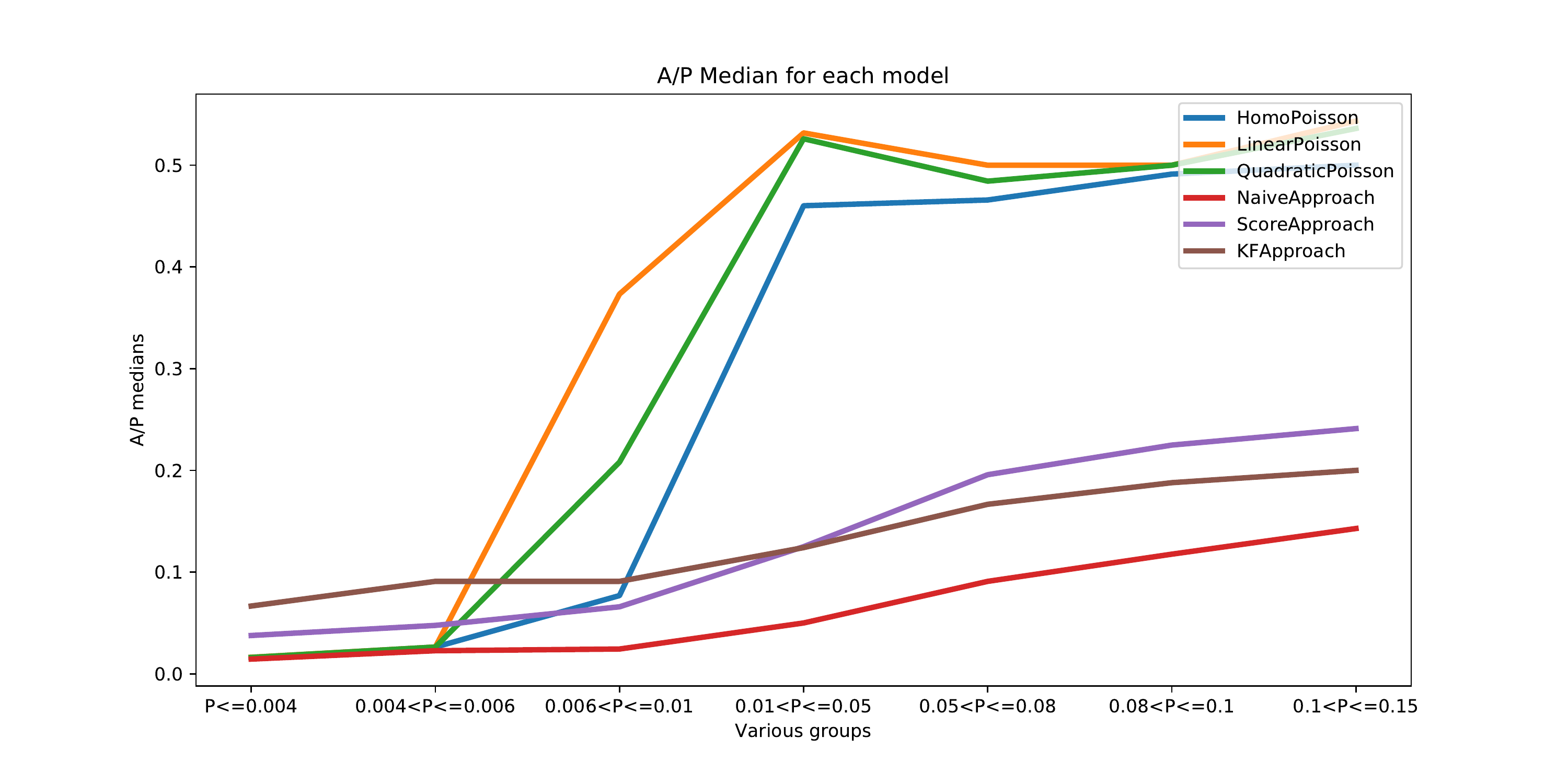}
\end{center}
\caption{Performance for A/P Medians for various models in each group. The plot shows the performance for each model with the degree of imbalanced dataset.}
\label{fig5}
\end{figure}

\section{Conclusion}
 An unsupervised approach based on a stochastic intensity model is investigated to detect fraud in imbalanced dataset. The Cox-Ingersoll-Ross (CIR) process is proposed with the advantage to guarantee a positive value for the fraud intensity. In this context, a closed form solution for the prediction probability of fraud is derived. Using the probability of fraud observed on the transactions, we have shown how to estimate the dynamic intensities by the Kalman-Filter method. Our methodology is applied to financial datasets. To evaluate the performance of our model, we consider in the paper other models of predicting fraud by the intensity-based approach.  These include the homogeneous Poisson process, the linear and quadratic inhomogeneous Poisson processes, a baseline approach and a random walk approach. All these models are compared to our model. We found that our Kalman filter approach outperforms the other approaches in the case of the more imbalanced dataset. When the fraud probability of the dataset increases, the performance of our model decreases. In this context, the linear intensity model is the better one following by the quadratic and the homogeneous Poisson process. Finally, all the models perform better than the baseline model. The main contributions of this paper are: 1. Our model is the first unsupervised approach for fraud detection using a stochastic intensity. It would be useful for datasets for which the fraud labels are not available. 2. The Cox-Ingersoll-Ross (CIR) process conducts to closed form solutions with few parameters and this greatly reduces the computational costs and the over-fitting. 3. Instead of using the observed fraud probability to estimate the intensity by the Kalman filter, the model could also be challenged by applying deep machine learning algorithms.  4. Our model is complete in the sense that it combines the information on the instantaneous rate of fraud with the fraud causality for the prediction of fraud events. So, the question of why and how the fraud occurs is investigated.
\bibliography{toto}
\bibliographystyle{icml2019}


\end{document}